\documentclass{llncs}

\newenvironment{proofof}[1]{\par\noindent{\it Proof of #1:}}{\qed}

\usepackage{amssymb}
\usepackage{graphicx}
\usepackage{epsfig}
\newcommand{\remove}[1]{}

\newcommand{\ignore}[1]{}%

\ignore{
\newtheorem{theorem}{Theorem}[section]
\newtheorem{lemma}[theorem]{Lemma}
\newtheorem{corollary}[theorem]{Corollary}
\newtheorem{proposition}[theorem]{Proposition}

\newtheorem{claim}[theorem]{Claim}
\newtheorem{definition}[theorem]{Definition}

}

\newcommand{\por}{p_{or}}
\newcommand{\ellora}{\ell_{or}^1}
\newcommand{\ellorb}{\ell_{or}^2}
\newcommand{\ellori}{\ell_{or}^i}
\newcommand{\bellora}{\bar{\ell}_{or}^1}

\newcommand{\For}{F_{or}}
\newcommand{\Fand}{F_{and}}

\newcommand{\bFor}{\widehat{F}_{or}}
\newcommand{\bFand}{\bar{F}_{and}}

\newcommand{\lowandi}{L_{and}^i}
\newcommand{\lowanda}{L_{and}^1}
\newcommand{\lowandb}{L_{and}^2}

\newcommand{\lowori}{L_{or}^i}
\newcommand{\lowora}{L_{or}^1}
\newcommand{\loworb}{L_{or}^2}

\begin{document}

%\title{Market Game
\title{The AND-OR game: Equilibrium Characterization \\
(working paper)}
\author{
Avinatan Hassidim\inst{1}
%\thanks{Google, Tel Aviv.}
%
\and Haim Kaplan\inst{2}
\ignore{\thanks{School of Computer science, Tel Aviv University.
This research was supported in part
 by a grant from the Israel
Science Foundation (ISF), by a grant from United States-Israel
Binational Science Foundation (BSF) and by The Israeli Centers of
Research Excellence (I-CORE) program, (Center  No. 4/11). .
Email:haimk@cs.tau.ac.il.}}
\and Yishay Mansour\inst{2} \ignore{\thanks{School of Computer
science, Tel Aviv University. Email:mansour@cs.tau.ac.il. This
research was supported in part by a grant from the the Science
Foundation (ISF), by a grant from United States-Israel Binational
Science Foundation (BSF), by a grant from the Israeli Ministry of
Science (MoS),  by The Israeli Centers of Research Excellence
(I-CORE) program, (Center No. 4/11) and by the Google
Inter-university center for Electronic Markets and Auctions. }}
\and Noam Nisan\inst{3} \ignore{\thanks{School of Computer Science
and Engineering, Hebrew University of Jerusalem. Supported by a
grant from the Israeli Science Foundation (ISF),  by The Israeli
Centers of Research Excellence (I-CORE) program, (Center  No. 4/11)
and by the Google Inter-university center for Electronic Markets and
Auctions.}}
% and by the Google Inter-university center for Electronic Markets and Auctions.
%
}

\institute{Bar-Ilan University and Google Inc. \and Tel Aviv
University
%\and Tel Aviv University
\and Microsoft Research and Hebrew University of
Jerusalem}

\maketitle

\begin{abstract}
We consider a simple simultaneous first price auction for multiple
items in a complete information setting. Our goal is to completely
characterize the mixed equilibria in this setting, for a simple, yet
highly interesting, {\tt AND}-{\tt OR} game, where one agent is
single minded and the other is unit demand.
%This game is especially interesting since it is a basic example for which Walrasian
%equilibrium does not exist.
%
\ignore{
 A simple {\tt AND}-Or games has two agents.
One agent, the {\tt AND}, wants both items and has value $1$ while
the other agent, the {\tt OR}, has value $v$. For $v\leq 1/2$ there
are simple equilibria
 We completely characterize the mixed equilibria of the {\tt AND}-{\tt OR} game, showing that
they are all slight variants of a single one.}
\end{abstract}

\section{Introduction}

\ignore{
% Background Walrasian Eq

Much of the economic theory challenges can be view as {\em how to
allocate scarce resources}. The main goal is to show that
appropriate pricing of the resources would lead to clearing the
markets, namely, supply and demand match. This is at the core of
many of the general equilibrium theory results, such as the famous
Arrow-Debrue model.

Unfortunately, there is a hidden assumption of convexity, an
assumption that breaks down when there are indivisible good. This is
not a minor deficiency, since many of the interesting motivations
for algorithmic game theory are indeed this case. This is the case
even if we assume that the agents are quasi linear with respect to
monetary transfer.
% Beyond Walrasian
}

Walrasian equilibrium is one of the most basic models in economic
theory. Items are priced in such a way that for each item either the
market clears (supply equals demand) or if there is an excess supply
it is priced at zero.
When there is a Walrasian equilibrium, it captures nicely the
``right'' pricing of items. Unfortunately, Walrasian equilibria are
guarantee to exists only for limited classes of agents'
valuations, namely gross-substitute valuations.

A different way of presenting the market is to
auction the items simultaneously, and analyze the resulting
equilibria. For simultaneous first price auction, the resulting pure
Nash equilibria are in one-to-one correspondence to the Walrasian
equilibria, with the same prices and allocations
\cite{HassidimKMN11,Bikhchandanil99}.
Considering the market as a simultaneous first price auction allows us
to consider it as a game, and study the
resulting equilibria. Fortunately, there is always a mixed
Nash equilibrium, with some tie breaking rule, and approximate Nash
equilibrium with any tie breaking rule \cite{HassidimKMN11}.

A typical example of a case where there are no Walrasian equilibria
is when there are two agents, one is single minded while the other
is unit demand. The {\tt AND}-{\tt OR} game is exactly this setting, with
two items. The {\tt AND} valuation is $1$ if it gets  both items and zero otherwise,
while the {\tt OR} has a value of $v$ for any single item (or both) and zero otherwise.
For $v>1/2$ there is no Walrsian equilibrium (or equivalently, pure
Nash equilibrium) and this is the interesting and challenging case
we focus on in this paper. A specific mixed Nash equilibrium for the {\tt
AND}-{\tt OR} game was presented in \cite{HassidimKMN11} (we
review it in Section~\ref{section-model--and-or}).

In this work we completely characterize the resulting mixed Nash equilibria
of the {\tt AND}-{\tt OR} game.
We show that while the equilibrium is technically not unique, it is
almost unique. More precisely we show that all mixed Nash equilibria
to the {\tt AND}-{\tt OR} game have the same marginal bid
distributions for each item. The resulting prices and
utilities of the {\tt AND} and {\tt OR} agents are the same in all mixed Nash equilibria.
We complement our characterization with a study of the properties of
these equilibria.

\ignore{
%YM -weak, should we delete?!
The issue of uniqueness and characterization of equilibria is many times central in game
theory. One of the reasons is it predictability. If we have multiple
equilibria, we can not be certain which equilibrium would be the
outcome, in contrast to a unique equilibrium has clearly solves this
problem. The characterization of the equilibria allows us to argue regarding all the possible equilibria.}

\smallskip
\noindent{\bf Related Work:}
There is a recent interest in algorithmic game theory to study
simple simultaneous and sequential auctions to allocate multiple
goods. Bhawalkar and Roughgarden \cite{BhawalkarR11} studied second
price simultaneous auctions and their Price of Anarchy (PoA). They
show that under the assumption of conservative bidding the PoA for
sub-modular valuations is a constant and for sub-additive valuations
it is logarithmic.

Hassidim et al. \cite{HassidimKMN11} studied a market which is based
on first price auctions. They show that the pure equilibria
correspond to Walarsian prices, and prove that a
mixed Nash equilibrium always exists\footnote{the issue is that the prices are
continuous.}. Similar to \cite{BhawalkarR11}, they show for
sub-modular valuations a constant PoA and for sub-additive
valuations a logarithmic PoA.

Leme at al. \cite{LemeST12soda} studied a sequential auctioning of
the items, where the solution concept is a sub-game perfect
equilibrium. They show that a pure sub-game perfect equilibrium
always exists, and that for unit demand buyers the PoA is at most
$2$ while for submodular buyers it might be unbounded.

Szentes \cite{Szentes07} studied a game with two identical bidders
and two identical items in the full information model, where both
bidders are either sub-additive or supper-additive. In addition to
simple pure Nash equilibria, he exhibits a family of symmetric mixed
Nash equilibria. This work was extended in \cite{SzentesR03a} to
three items (where each of the two agent desires at least two of the items)
and in \cite{SzentesR03} to multiple items and agents
(where again, each agent desires a strict majority of the items).
All the above works exhibit specific symmetric mixed equilibria,
and none of them address the characterization of all mixed Nash equilibria.

\section{The {\tt AND}-{\tt OR} Game: Model and preliminaries} \label{section-model--and-or}

We have two players an {\tt AND} player and {\tt {\tt OR}} player.
The {\tt AND} player has a value of $1$ if he gets both the items in
$M=\{1,2\}$, and the {\tt OR} player has a value of $v$ if she gets
any item in $M$. Formally, $v_{and}(M)=1$ and for $S\neq M$ we have
$v_{and}(S)=0$, also, $v_{or}(T)=v$ for $T\neq\emptyset$ and
$v_{or}(\emptyset)=0$. Both players have a quasi-linear utility with
respect to money, i.e., getting subset $S$ and paying price $p$ has
a utility of $u(S)=v(S)-p$, and are risk neutral. We assume a full
information setting, namely both players know each other's
valuation.

In the {\tt AND}-{\tt OR} game both players participate in two
simultaneous first price auctions, one for each item.
Namely, each player places bids $(x_1,x_2)$ where $x_i$ is the bid
on item $i$. The highest bidder on each item wins it and pays its
bid. We will denote by $H$ the maximum allowed bid in the game, so
$0 \le x,y \le H$. (Clearly it suffices to have $H \le max(1,v)$.)
%
% Namely, each submits a bid of each item, the highest bid wins and pays its bid.
%
Completely specifying a first price auction requires a {\em tie
breaking rule}, which specifies the winner in case of identical
bids.
%, and the rule may be randomized.
%
We make the assumption that the tie breaking rule only depends on
the bids for the tied item (and not on the bids on the other item),
and allow for a randomized tie breaking rule.

%{\bf HK We do not say where we use this assumption}

When $v\leq 1/2$ there is a Walresian equilibrium for any price
$p\in[v,1/2]$ per item. This implies a pure Nash Equilibrium in
which both players bid $p$ on each item, and the {\tt AND} player
wins both items (assuming the tie breaking rule favors {\tt AND}).
For this reason we are interested in the case when
$v>1/2$. It is easy to verify that in this case is no Walresian
equilibrium. For completeness we show that there is no pure Nash
equilibrium.

\begin{claim}
There is no pure Nash equilibrium in the {\tt AND}-{\tt OR} game.
\end{claim}

\begin{proof} Assume for contradiction there was a pure Nash
equilibrium with some tie breaking rule. Let the {\tt AND} bid
$(x,y)$.
%W.l.o.g. assume that $x\leq y$, and either $x+y\leq1$ or $x+y>1$.
%
Since the {\tt AND} value is $1$, if $x+y>1$ the {\tt AND} has a
negative utility for any best response of the {\tt OR}. In the case
that $x+y\leq 1$, the best response of the {\tt OR} is to out bid
the lower bid of the {\tt AND}.
%Since the lower bid is at most $1/2$ the {\tt OR} has a strictly positive utility.
Therefore, {\tt AND} will have a non-negative utility only if it bids
$(0,0)$. In case that {\tt AND} bids $(0,0)$, the {\tt OR} wins one item at a price of at
most $\epsilon$. But then the {\tt AND} can deviate and bid
$(2\epsilon, 2\epsilon)$ and have a positive utility. Contradiction. \qed
\end{proof}

Next, we describe the mixed Nash equilibrium from
\cite{HassidimKMN11}.
\begin{itemize}

\item The {\tt {\tt AND}} player bids $(y,y)$ where $0 \le y \le 1/2$ according to cumulative distribution
$\Fand^*(y)=(v-1/2)/(v-y)$  (where $\Fand^*(y)=Pr[bid \le y]$). In
particular, There is an atom at 0: $Pr[y=0]=1-1/(2v)$.

\item The {\tt OR} player bids $(x,0)$ with probability 1/2 and $(0,x)$ with probability 1/2,
where $0 \le x \le 1/2$ is distributed according to cumulative
distribution  $\For^*(x)=x/(1-x)$.

\end{itemize}

Note that since the {\tt OR} player does not have any mass points in
his distribution, the equilibrium holds for any tie breaking
rule. The proof that this is indeed an equilibrium is in
\cite{HassidimKMN11}. The main goal of this paper is to characterize
the mixed Nash equilibria of the {\tt AND}-{\tt OR} game, and to
show that this is ``essentially'' the only mixed Nash equilibrium.

\section{Characterization of the mixed Nash equilibria}

%\subsection{Main result and outline of the proof}

The following theorem is our main result which characterizes the
mixed Nash equilibria of the {\tt AND}-{\tt OR} game. The {\tt OR}
player has to play the mixed strategy $\For^*$. The {\tt AND} player
can play various mixed strategies, but their marginal bid
distribution on each item is identical to $\Fand^*$, and the
probability mass at $(0,0)$ is the same as of $\Fand^*$. While there
is more than a single equilibrium, they all have ``essentially'' the
same outcomes, i.e., the same expected utilities, payments, and
allocation probabilities.

\begin{theorem}
\label{thm:main}
 A pair of strategies $(\Fand,\For)$ is an
equilibrium of the {\tt AND}-{\tt OR} game if and only if
\begin{enumerate}
\item {\tt OR}'s strategy is
$\For=\For^*$.
\item {\tt AND}'s strategy has the same marginal distributions of the bids for each item as
$\Fand^*$:  $\Fand(x,H)=\Fand^*(x,H)$ and $\Fand(H,y)=\Fand^*(H,y)$
for all $x,y$, and the same probability of $(0,0)$:
$\Fand(0,0)=\Fand^*(0,0)$.
%%%, and for any $x,y>0$, $\Fand(x,0)=\Fand(0,0)=\Fand(0,y)$.
%Also, $\Fand$ is weakly dominated as a strategy by $\Fand^*$.
\end{enumerate}
%\item
Furthermore, $\Fand$ is weakly dominated as a strategy by $\Fand^*$  and the following quantities are the same as in the equilibrium
$(\Fand^*,\For^*)$:
%\begin{enumerate}
%\item
(1) The allocation probabilities, i.e. the probabilities of each
player winning each bundle.
%\item
(2) The expected payments, and utilities of each player.
Thus also the expected revenue and social welfare.
%\end{enumerate}
\end{theorem}

The proof of the above theorem is quite involved, and most of the
paper is devoted to it. In the following we give a high level view
of the proof.
%
%\begin{enumerate}
%\item
We start with preliminary set-up. In Section \ref{s:interval} we
show that the support of the bids of each player for each item is
(essentially) an interval, with both players having the same upper
bound.
%\item
In Section \ref{or-positive} we observe that  {\tt OR} must get non-zero utility.
%\item
%
In Section \ref{sim} we show that without loss of generality we can
increase the correlation between the bids of {\tt AND} on the two
items without changing its marginal distribution on each item. This
allows us to consider maximally correlated distributions for {\tt
AND}. (We formally define in Section~\ref{sim} what we mean by
maximally correlated.)
%\item
In Section \ref{or-axis} we prove that the bids of {\tt OR} must be
on the axis: i.e. either $(0,x)$ or $(x,0)$; and then in Section
\ref{and-diagonal} we show that {\tt AND} must bid on the diagonal,
i.e., $(x,x)$.
%\item
At this point we derive, in Section \ref{exact} the exact form of
the equilibrium distributions.  Section \ref{thm-pf} puts everything
back together showing how the main theorem is implied.
%\end{enumerate}

\subsection{The Bids on Each Item form an Interval}\label{s:interval}

The results in this section apply to any equilibrium in a market
with two bidders and any number of items. We still state these
results using the names {\tt AND} and {\tt OR} for the players but
we only assume that $(\Fand,\For)$ are an equilibrium of some market
game. We do not  rely in this section on the specific form of the
utilities of {\tt AND} and {\tt OR}, but only on the fact that our
game is a simultaneous first price auction. This implies that
Lemma~\ref{interval} also holds with the roles of {\tt AND} and {\tt
OR} reversed.

\begin{lemma}
\label{lemma:high_bid}
The highest bids of the two players on a
particular item are equal.  I.e. $\Fand(x,H)=1$ if and only if
$\For(x,H)=1$ (and similarly for the other items.)
\end{lemma}

\begin{proof}
Assume that $\For(x_0,H)=1$ for some $x_0$. For any bid $(x,y)$ of
{\tt AND}, with $x>x_0$, {\tt AND} gets strictly lower utility than
$(x_0+\epsilon,y)$ (for any $\epsilon<x-x_0$) since it wins in
exactly the same cases $(x,y)$ won, and always pays strictly less.
It follows that no bid with $x>x_0$ is a best-response for {\tt
AND}, thus {\tt AND} always bids at most $x_0$ on the first item,
i.e., $\Fand(x_0,H)=1$. The other direction is analogous. \qed
\end{proof}

Given Lemma \ref{lemma:high_bid} it is possible to define the
``highest bid'' on an item: $h_1=\min\{x\mid \Fand(x,H)=1\}=\min\{x
\mid \For(x,H)=1\}$ and $h_2=\min\{y\mid \Fand(H,y)=1\}=\min\{y\mid
\For(H,y)=1\}$ (where the minimum is actually achieved due to the
right-continuity of CDFs.) We will argue below that in fact
$h_1=h_2$.

\begin{lemma} \label{interval}
Let $0<b<c$ such that {\tt OR} never bids between $b$ and $c$ on an
item, i.e., $\For(b,H)=\For(c,H)$ and such that both players
sometimes bid at most $b$, i.e., $\For(b,H)>0$ and $\Fand(b,H)>0$.
Then both players always bid at most $b$, i.e.
$\For(b,H)=\Fand(b,H)=1$.
\end{lemma}

\begin{proof}
First, if {\tt OR} never bids above $c$, i.e., $\For(c,H)=1$, then
also $\For(b,H)=\For(c,H)=1$ and by Lemma~\ref{lemma:high_bid} also
$\Fand(b,H)=1$.

Otherwise, for a contradiction, assume {\tt OR} bids above $c$.
Define $d$ to be the infimum bid of {\tt OR} above $b$: $d=\inf \{x
\mid \For(x,H) > \For(b,H)\}$. This implies that for any
$\epsilon>0$ the {\tt OR} bids in $[d,d+\epsilon)$ with positive
probability. We show a sequence of properties that
depends on $d$:\\
%\begin{enumerate}
%\item
%
{\bf (I)} {\tt AND} does not bid in the range $(b,d)$ on that item
(i.e., $\Fand(x,H)=\Fand(b,H)$ for all $b \le x<d$). Assume for a
contradiction that {\tt AND} bids $(x,y)$ for some $b<x<d$. Now
consider a deviation of {\tt AND} that bids $(b+\epsilon,y)$. Both
$(x,y)$ and $(b+\epsilon,y)$  win in exactly in the same cases,
however, with the bid $(b+\epsilon,y)$ {\tt AND} pays strictly less
(for any $\epsilon< x-b$). Since this occurs with positive
probability, we have a contradiction that $(x,y)$ is a best
response of {\tt AND}.\\
%
%\item
{\bf (II)} Assume {\tt AND} player does not have an atom at $d$.
%(i.e. $\Fand(d,H) = \sup \{ \Fand(x,H) : x<d\} = \Fand(b,H)$).
%
Consider a deviation of the {\tt OR} player switching any bid for the
first item which is in the range $(d,d+\epsilon)$ (for a small
enough $\epsilon \ge 0$), with the bid $(b+d)/2$.
The probability that {\tt AND} bids in $(d,d+\epsilon)$ goes to zero as $\epsilon$
goes to zero, since {\tt AND} does not have an atom at $d$.
This upper bounds the loss of {\tt OR} in the deviation.
On the other hand, the payments decrease by at least $(d-b)/2 >0$ and this happens with
positive probability (at least the probability that {\tt AND} bids below $b$).
Therefore we reached a contradiction to the assumption that
{\tt OR} is best responding.\\
%\item
{\bf (III)} Assume {\tt AND} has an atom at $d$ and the {\tt OR}
does not have an atom at $d$.
%Exactly the same argument as before shows a contradiction.
Consider a deviation where the {\tt AND} switches the bid of $d$ for the
first item by the bid $(b+d)/2$. Since the {\tt  OR} does not have an atom at $d$,
the probability that {\tt AND} wins does not change, and the payments go down by
$(d-b)/2$ with constant probability (at least the probability that {\tt OR} bids below $b$.)
Therefore we reached a contradiction to the assumption that
{\tt AND} is best responding.\\
{\bf (IV)} Assume both {\tt AND} and {\tt OR} have an atom at $d$.
Now, look at the tie breaking rule at $d$, it gives {\tt AND } probability $q_d$ of wining the tie at $d$
and {\tt OR} a probability of $1-q_d$ winning.
At least one of the players does not always win.
\ignore{\footnote{Here we are using the fact
that the tie breaking rule does not depend on the bids of the second
item.} }
That player may want to increase its bid to $d+\epsilon$ and always
win the item -- this will be strictly beneficial unless its expected
utility from bidding $d$ is exactly zero.
But in that case reducing its bid to $b+\epsilon$ will strictly
increase its utility: winning whenever he previously did and paying
less (again, winning the item with these bids has positive
probability.)
%\end{enumerate}

We reached a contradiction to the assumption that {\tt OR} bids
above $c$, therefore $\For(b,H)=1$. By Lemma~\ref{lemma:high_bid} we
have that $\Fand(b,H)=1$.
 \qed
\end{proof}

\subsection{OR Gets Positive Utility}\label{or-positive}

\begin{lemma}
\label{lem:or-positive} In any mixed Nash equilibrium $(\Fand,\For)$
the expected utility of the {\tt OR} player is strictly positive.
\end{lemma}

\begin{proof}
For contradiction assume that the expected utility of the {\tt OR}
player is zero, i.e., $u_{or}(\Fand,\For)=0$. Consider the following
deviation of the {\tt OR} player. Let $\eta=(0.5+v)/2 >1/2$. The
{\tt OR} player bids with probability $1/2$ the bids $(0,\eta)$ and
with probability $1/2$ the bids $(\eta,0)$. Since we are at an
equilibrium the expected utility of this deviation is $0$. This
could happen only if
 the
{\tt AND} player always bids in $\Fand$ above $\eta$ for both items,
i.e., $\Fand(\eta,H) = \Fand(H,\eta) = 0$.

If the {\tt AND} player  always bids above $\eta$ for both items
then he has a non-negative utility only if he always loses both
items.
It follows that the {\tt OR} player, using $\For$, always wins both
items with a cost of at least $\eta$ for each. However, this implies
that the {\tt OR} player pays at least $v+0.5$. Since the value of
the {\tt OR} player is $v$ she has a negative utility, which is a
contradiction. \qed
\end{proof}

Let $\por(x,y)$ be the probability that the {\tt OR} wins at least
one item with the bid $(x,y)$, i.e.,
$\por(x,y)=\Fand(x,H)+\Fand(H,y)-\Fand(x,y)$. The following simple
corollary to Lemma~\ref{lem:or-positive} would be useful.

\begin{corollary}
\label{cor:or-positive} In any mixed Nash equilibrium
$(\Fand,\For)$, if $\For(x,y)>0$ then $\por(x,y)>0$.
\end{corollary}

%{\bf HK its seems that this corollary is currently not cited. Maybe
%at some places where we cite the lemma we want to cite the corollary
%instead.}

Consider the highest bids $h_1$ and $h_2$. Clearly if $h_1 > h_2$
and $h_2 < v$ then the {\tt OR} can gain by deviating and bidding
$h_2 +\epsilon$ on the second item and $0$ on the first item
 whenever it used to bid $y
> h_2$ on the first item. If $h_1 > h_2=v$  then {\tt OR} cannot
have positive utility contradicting Lemma \ref{lem:or-positive}. So
we have the following corollary.
\begin{corollary} \label{cor:highestbid}
The highest bid of the players on the first item ($h_1$) is equal to
the highest bid of the players on the second item ($h_2$).
\end{corollary}

We let $h=h_1=h_2$.

\subsection{Identical Marginal Distributions and Correlation}
\label{sim}

It will be convenient to consider the projection of the joint bid
distribution of a player on the individual coordinates, i.e., the
single items.

%Similarity == Identical Marginal Distributions
\begin{definition}
Two CDF's $F$ and $F'$ are called {\em  Identical Marginal
Distributions} if their marginals are identical. I.e.,
$F(x,H)=F'(x,H)$ for all $x$ and
$F(H,y)=F'(H,y)$ for all $y$.  % Same for $G$ and $G'$.
\end{definition}

The following proposition takes advantage of the fact that the
decision of each auction depends only on the marginal distribution.

\begin{proposition} \label{equiv}
If $\Fand$ and $\Fand'$ are Identical Marginal Distributions then
against any strategy $\For$:
\begin{itemize}
\item Each item is won by each player with the same probability in $(\Fand,\For)$ and $(\Fand',\For)$.
\item The expected payments of each player are the same in $(\Fand,\For)$ and $(\Fand',\For)$.
\end{itemize}
%Similarly if $G$ and $G'$ are Identical Marginal Distributions and playing against $F$.
\end{proposition}

Note that the above proposition states that the probabilities of
winning any single item by any player are identical in
$(\Fand,\For)$ and $(\Fand',\For)$, but the probability of winning
both items might differ.

\begin{proposition} \label{better}
If $\Fand$ and $\Fand'$ are Identical Marginal Distributions then,
for every $\For$, the following conditions are equivalent:
\begin{enumerate}
\item
$u_{and}(\Fand,\For) \le u_{and}(\Fand',\For)$.
\item
$u_{or}(\Fand,\For) \ge u_{or}(\Fand',\For)$.
\item
$Pr_{(\Fand,\For)}[\textrm{{\tt AND} wins  both items}] \le
Pr_{(\Fand',\For)}[\textrm{{\tt AND} wins  both items}]$.
\item
$Pr_{(\Fand,\For)}[\textrm{{\tt OR} wins an item}] \ge Pr_{(\Fand',\For)}[\textrm{{\tt OR} wins an item}]$
\item
$Pr_{(\Fand,\For)}[\textrm{{\tt AND} wins  no items}] \le
Pr_{(\Fand',\For)}[\textrm{{\tt AND} wins  no items}]$.
\end{enumerate}
\end{proposition}

\begin{proof}
Since $\Fand$ and $\Fand'$ are Identical Marginal Distributions  the
expected payments are identical for both players.
This implies that the only parameter that influences the utility is
the probability of winning. (For the {\tt AND} player, winning both
items, for the {\tt OR} player, winning one of the two items.) Thus
(1) is equivalent to (3) and (2) is equivalent to (4). But notice
that  always  exactly one of {\tt AND} or {\tt OR} wins, and
therefore if the probability that {\tt AND} wins increases, then the
probability that {\tt OR} wins decreases. It follows that (3) and
(4) are equivalent. Finally, since $\Fand$ and $\Fand'$ are
Identical Marginal Distributions the players win each item with the
same probability. Let $p_1$ be the probability that {\tt AND} wins
item 1 and $p_2$ be the probability that {\tt AND} wins item 2 (both
with $\Fand$ and with $\Fand'$!).  Then we have $Pr[{\tt
AND}\:wins\:no\:item] = 1 - p_1 - p_2 + Pr[{\tt
AND}\:wins\:both\:items]$ and thus (3) and (5) are equivalent.
\qed\end{proof}

An immediate corollary is that if the utilities of one player are identical,
then the utilities of the other player are also identical.

\begin{corollary}
\label{cor:eq} Assume that $\Fand$ and $\Fand'$ are Identical
Marginal Distributions. Then $u_{and}(\Fand,\For) =
u_{and}(\Fand',\For)$ iff $u_{or}(\Fand,\For) = u_{or}(\Fand',\For)$
\end{corollary}

A very important building block in our proof is the notion of {\em
maximally correlated}. Intuitively, if the support of a distribution
is on a monotone increasing line, then this distribution is
maximally correlated. Since we want to show that the {\tt AND}
player support is essentially the diagonal, this would be very
useful to characterize its bid distribution.

\begin{definition}
For a CDF $\Fand$ %and $G$,
define $\bFand(x,y)=\min(\Fand(x,H),\Fand(H,y))$.
%and $\bFor$ is
%defined by $\bFor(x,y)=\max(0,\For(x,H)+\For(H,y)-1)$.
$\Fand$ is called {\em maximally correlated} if $\Fand=\bFand$.
%and $\For$ is called
%maximally-anti-correlated if $\For=\bFor$.
\end{definition}

Note that $\bFand(x,H)=\min(\Fand(x,H),\Fand(H,H)) =
\min(\Fand(x,H),1)=\Fand(x,H)$, and therefore the following
proposition holds.

\begin{proposition} \label{bf-sim}
Every CDF $\Fand$ %and $\For$,
%we have that $\bFand$ is a CDF
is Identical Marginal Distribution to $\bFand$.
% and $\bFor$ is a CDF similar to $\For$.
\end{proposition}

The following proposition claims that $\bFand$ stochastically
dominates $\Fand$.

\begin{proposition}
\label{prop:stoch-dominance}
For every $(x,y)$, $\Fand(x,y) \le \bFand(x,y)$. I.e. $\Fand$
stochastically dominates $\bFand$.
%For every $(a,b)$, $\For(a,b) \ge
%\bFor(a,b)$.  I.e, $\bFor$ stochastically dominates $\For$.
\end{proposition}

The following lemma shows that a maximally correlated strategy
$\bFand$ weakly dominates the original strategy $\Fand$ .

\begin{lemma} \label{weak-dom}
Every $\Fand$ is weakly dominated, as a strategy of {\tt AND}, by
$\bFand$.
%Every $\For$ is weakly dominated, as a strategy of OR, by $\bFor$.
\end{lemma}

\begin{proof} \label{dom}
Fix a pure bid $(x,y)$ of {\tt OR} and let $\For$ bid $(x,y)$ with
probability 1.  By Proposition \ref{better} we have that
$u_{and}(\Fand,\For) \le u_{and}(\bFand,\For)$ if and only if
$\Fand(x,y)=Pr_{(\Fand,\For)}[\mbox{{\tt AND} wins both items}] \le
Pr_{(\bFand,\For)}[\mbox{{\tt AND} wins both items}]=\bFand(x,y)$.
The latter holds, since by Proposition~\ref{prop:stoch-dominance} we
have $\Fand(x,y)\le \bFand(x,y)$.
%since $\Fand$ stochastically dominates $\bFand$ and
%thus the probability of bidding at most $x$ on the first item and at
%most $y$ on the second is smaller(or equal) in $\Fand$.
\qed\end{proof}

The following lemma shows that if we have an equilibrium
$(\Fand,\For)$, and we  replace $\Fand$ by $\bFand$, then we still
remain in an equilibrium. The main part of the proof is showing that
the {\tt OR} strategy remains a best response to $\bFand$.

\begin{lemma} \label{eqtoeq}
If $(\Fand,\For)$ is a Nash equilibrium then $(\bFand,\For)$
%and $(\Fand, \bFor)$ are each
is a Nash equilibrium.  Moreover it produces exactly the same
distribution on allocations, payments, and  utilities as does
$(\Fand,\For)$.
\end{lemma}

\begin{proof}
By Lemma~\ref{weak-dom}, $\bFand$ dominates $\Fand$ so
$u_{and}(\Fand,\For)\leq u_{and}(\bFand,\For)$. Since $(\Fand,\For)$
is an equilibrium, then $u_{and}(\Fand,\For) = u_{and}(\bFand,\For)$
and by Corollary~\ref{cor:eq} we also have that $u_{or}(\Fand,\For)
= u_{or}(\bFand,\For)$. Since $\For$ is best response to $\Fand$ we
have that $u_{or}(\Fand,\For) \geq u_{or}(\Fand,\For')$, for any
$\For'$.

Since by Lemma~\ref{weak-dom} we also have that
$u_{and}(\bFand,\For') \geq u_{and}(\Fand,\For')$ then by
Proposition~\ref{better}, $u_{or}(\bFand,\For') \leq
u_{or}(\Fand,\For')$. Thus,
\[
u_{or}(\bFand,\For') \leq u_{or}(\Fand,\For') \leq
u_{or}(\Fand,\For) = u_{or}(\bFand,\For),
\]
which implies that $\For$ is a best response to $\bFand$.
From Proposition~\ref{better} we get that the allocations, payments, and  utilities are identical.
\qed\end{proof}

We will continue the analysis assuming that $\Fand=\bFand$ is
already maximally correlated and will derive the form of the
equilibrium under this assumption. By Lemma \ref{eqtoeq} this
characterization will then apply to other distributions that are
Identical Marginal Distributions to $\Fand$, and form an equilibrium
with $\For$. We will keep the notation $\bFand$ to stress that it is
maximally correlated.

%It follows that if we are only interested in equilibria in
%un-dominated strategies then {\tt AND} plays a maximally correlated
%$\bFand$
% and OR plays a maximally anti-correlated $\bFor$.
%Even if we are interested in all equilibria, their outcome will be
%equivalent to that obtained by AND playing a maximally correlated
%$\Fand$.  So it is enough to analyze this case, as we will do in the rest of the proof.

%We could also do the dual thing for OR, and replace $\For$ by
%a``maximally anti-correlated''
%$\bFor(x,y)=max(0,\For(x,H)+\For(H,y)-1)$, but we will not require
%this for the rest of our analysis.

%[[YM: Does this imply that the AND plays a monotone line ??? ]]

%Assume that the AND player distribution $\bFand$ is maximally correlated. Then the support is
%$\{(x,y): \Fand(x,H)=\Fand(H,y) \}$

\ignore{ {\bf Comment:} In a similar way that we did for {\tt AND}'s
strategies, we could show that {\tt OR}'s strategy is dominated (as
a strategy) by a similar maximally `` anti-correlated''
$\bFor(x,y)=\max(0,\For(x,H)+\For(H,y)-1)$. We will not be needing
this though and will actually show that {\tt OR} must be exactly of
a certain form. [[YM: Do we need this?!]]}

\subsection{OR Bids on the Axis}\label{or-axis}

%In this section we are considering an equilibrium $(\Fand,\For)$ of
%the {\tt AND}-{\tt OR} game with {\tt OR}'s value being $v>1/2$, and
%with $\Fand$ maximally correlated.

Our approach here is to first show that {\tt OR} must always place a
certainly-loosing bid on one of the items; and then to show that
{\tt AND} bids arbitrarily low on each item, implying that
certainly-loosing bids of {\tt OR} must be 0.

We start by defining the low bids for {\tt AND} in a distribution
$\Fand$:

\begin{definition}
$\lowanda = inf \{ x \mid \Fand(x,H)>0 \}$ and $\lowandb = inf \{ y
\mid \Fand(H,y)>0 \}$.
\end{definition}

Definition~\ref{def:lowor} specifies similar quantities for {\tt OR} but in a {\em
different} way.

\begin{definition}
Given a distribution of {\tt AND}, $\Fand$, we say that a bid $x$ of
{\tt OR} for item $i$ is low, if either $x=0$ or {\tt OR} with bid
$x$ always looses  item $i$ to {\tt AND}. We denote a low bid of
{\tt OR} by $\ellori$.
\end{definition}

%Note that any value $x<\lowanda$ is a low value for {\tt OR}.
By definition, a bid $x > 0$ of {\tt OR} is low if $x < \lowandi$ or if $x=\lowandi$ and
either {\tt AND} doesn't have an atom at $\lowandi$ or {\tt AND}
always wins the tie for item $i$ at $\lowandi$.

\begin{lemma}
\label{lem:or-low} Assume that $(\bFand,\For)$ is an equilibrium,
then the {\tt OR} player bids a low bid on exactly one of the items.
%(Formally,[[YM: Need to fix]] if $\For(x,y)>0$ then either
%$\For(x,\ellorb)=\For(x,y)$ or $\For(\ellora,y)=\For(x,y)$, but not both.)
Moreover the probability of bidding low on each one of the
items is positive.
\end{lemma}

\begin{proof}
Since $\bFand$ is maximally correlated, we have that the probability
that the {\tt OR} wins at least one item using a bid $(x,y)$ is
$\max\{\bFand(x,H),\bFand(H,y)\}$ (this holds since the probability
of winning at least one item is
$\por(x,y)=\bFand(x,H)+\bFand(H,y)-\bFand(x,y)$, and since $\bFand$
is maximum correlated $\bFand(x,y) =
\min\{\bFand(x,H),\bFand(H,y)\}$). Assume that $\bFand(x,H)\geq
\bFand(H,y)$. Consider a deviation of the {\tt OR} where she bids
$(x,0)$ instead of $(x,y)$. The probability that {\tt OR} wins at
least one item is the same, i.e., $\por(x,y)=\por(x,0)$.
The payment decreases % from $x\bFand(x,H)+y\bFand(H,y)$ to $x\bFand(x,H)$.
by $y$ times the probability that {\tt OR} wins the second item.
This is a strict decrease unless $y$ is low.  This establishes that
in each bid $(x,y)$ in the support of $\For$ at least one of $x$ or
$y$ must be low.

We now want to show that $\For(\ellora,\ellorb)=0$ for any low bids
$\ellora$ and $\ellorb$. Namely, the probability that the {\tt OR}
has both bids low is zero. Assume by contradiction that
$\For(\ellora,\ellorb)> 0$ for some low bids $\ellora$ and $\ellorb$
and consider the
following cases where {\tt OR} bids $(x,y)$ such that $x\leq \ellora$ and $y\leq \ellorb$:\\
{\bf (1)} {\tt OR} always looses both items. This implies that the
{\tt OR}  has
zero utility, contradicting Lemma~\ref{lem:or-positive}. \\
{\bf (2)} {\tt OR} wins with positive probability. Since {\tt OR}
always looses with a low value which is not zero we may assume that
either $\For(\ellora,0)> 0$ or $\For(0,\ellorb)> 0$. We assume that
$\For(\ellora,0)> 0$, the case where $\For(0,\ellorb)> 0$ is
symmetric. We also denote the probability that  {\tt OR} wins
 with a bid $(x,0)$ where $x\le \ellora$ by $p=\por(\ellora,0)
> 0$. Consider now the a deviation for the {\tt AND} player, in which
it increases every bid on every item by $\epsilon < p/3$. This
increases its probability of winning (and its expected welfare) by
at least $p$, and increases its payments by $2p/3$, which gives a
net utility gain of at least $p/3$.

Therefore, we have established that the {\tt OR} always bids one low
value and never bids both values low.

It remains to show that {\tt OR} must place a low bid with positive
probability on each item. By contradiction, assume that the {\tt OR}
always bids low only on one of the items, say item $1$. We have two cases:\\
{\bf (1)} The {\tt AND} player always wins item $1$. Let $\bellora$
be the supermum of the (low) bids of  {\tt OR}  on item $1$. Then it
is clear that in equilibrium the {\tt AND} never bids above
$\bellora+\epsilon$ on item $1$,
%and never bids below $\ellora-\epsilon$,
for any $\epsilon>0$. Now consider item $2$.
%In this case the price of item $2$ will always be $\min\{1,v\}$:
We have two cases depending on the relation between  $v$ and $1$:
(1a) If $v<1$ then (given that the {\tt AND} always wins item $1$
and no-matter in what price) the {\tt AND} player will win item $2$
and pay at most $v+\epsilon$. Since the {\tt AND} player  wins both
items the {\tt OR} player has  zero utility, contradicting
Lemma~\ref{lem:or-positive}. (1b) If $v>1$ then the {\tt OR} player
always wins item $2$. This implies that the {\tt AND} player always
loses, and has negative utility unless $\bellora=0$. However, in
this case the {\tt OR} player can deviate and bid $(2\epsilon,0)$
and always win item $1$, contradicting the assumption that we
have an equilibrium.\\
{\bf (2)} The {\tt OR} player wins item $1$ with positive
probability $p > 0$. This can happen only if $\bellora=0$. Let
$\delta$ denote the expected price of item $2$. If $\delta = 0$, the
{\tt AND} player can strictly increase its utility by bidding $(p/3,
p/3)$. If $\delta > 0$, then the auction has expected revenue at
least $\delta$, and therefore the expected utility of the {\tt OR}
player is at most $v - \eta$ for some $\eta > 0$. Since the {\tt OR}
player is always bidding $0$ on item $1$ the {\tt AND} player will
bid at most $\epsilon=\eta/3$ on item $1$ (this is true for any
$\epsilon
> 0$). But then bidding $(2\eta/3, 0)$ is a profitable deviation for
the {\tt OR} player.

%Since the {\tt
%OR} player is always bidding $0$ on item $1$ the {\tt AND} player
%will bid at most $\epsilon$ on item $1$. This implies that the {\tt
%OR} has a deviation $(2\epsilon,0)$ which is guaranteed to win the
%first item, and has utility $v-2\epsilon$. Therefore the {\tt OR}
%player wins some item with probability at least $1-2\epsilon/v\geq
%1-4\epsilon$. This implies that the {\tt AND} player's probability
%of winning is at most $4\epsilon$, which also bounds its expected
%utility and expected payments. Now consider the player that has a
%higher probability of winning item $2$:\\
% {\bf (2a)} If the {\tt OR}
%player wins item $2$ with probability at least half, its expected
%payment when winning this item is at most $4\epsilon$. So if the
%{\tt AND} player bids $(\epsilon,8\epsilon)$ it will win the second
%item with
% probability  at least quarter, and have utility almost
%quarter, compared to at most $4\epsilon$ that he has before
%deviating. This implies that the {\tt OR} cannot win item $2$ with
%probability at
%least half.\\
%
%{\bf (2b)} If the {\tt AND} player wins item $2$ with probability at
%least half. Consider a deviation of the {\tt AND} player that always
%bids  $\epsilon$ on the first item. This will guarantee that it will
%always win the first item, and have expected payment $5\epsilon$,
%hence utility of at least $1/2-5\epsilon$. The {\tt AND} player can
%have a utility of at least $1/2-5\epsilon$ only if it wins both
%items with that probability. But then the {\tt OR} cannot win item
%$1$ with probability at least $1-4\epsilon$.

We have established that the {\tt OR} player has to bid a low bid
with a positive probability on each item.
\qed\end{proof}

Now let us define for {\tt OR} $\lowori$, which is different from the definition of {\tt AND}.

\begin{definition}
\label{def:lowor} $\lowora = \inf \{ x \mid \For(x,y)>0 \:\mbox{and
y is low}\}$ and $\loworb = \inf \{ y \mid \For(x,y)>0 \:\mbox{and x
is low}\}$.
\end{definition}

\begin{lemma}
\label{lemma:atom_both} Assume that $(\bFand,\For)$ is an
equilibrium and $\lowanda = \lowora=l$. Then at most one of the
players can have an atom at $l$ in the marginal distribution of item
$1$.
\end{lemma}

\begin{proof}
Assume by  contradiction that they both have atoms at $l$.
%i.e., $\Fand(l,H)>0$ or $\For(\bellora,y)>0$.
By Lemma~\ref{lem:or-low}, the {\tt OR} player never has both bids
low, so $\lowora$ cannot be low. Since $\lowora$ is not low, the
{\tt OR} wins the tie with non-zero probability.

Now consider the following cases depending on $\lowanda$:\\
(1) If $\lowanda < 1$ then {\tt AND} can gain by always increasing
both its bids by $\epsilon$: The payments increases by at most
$2\epsilon$ but the winning probability increases by a constant.
({\tt AND} now wins the atom.)\\
(2) Assume $\lowanda \ge 1$. The {\tt AND} must  always bid $0$ on
item $2$, since if it bids higher, when {\tt OR} bids low on $2$
(which happens with positive probability by Lemma \ref{lem:or-low})
{\tt AND} has negative utility. In this case {\tt OR} gains by
bidding $\epsilon$ on item $2$ and $0$ on item $1$. It will always
win item $2$ and pay $\epsilon$ whereas previously {\tt OR} paid at
least $1$ when it won item $1$, which occurred with constant
probability.
\qed\end{proof}

We now show that the low values of {\tt AND} are indeed zero.

\begin{lemma} \label{lem:lowand}
$\lowanda = \lowandb = 0$.
\end{lemma}

\begin{proof}
Assume by way of contradiction that $\lowanda>0$,
and let us look at the relation between $\lowandb$ and $\loworb$.
By the previous lemma if $\lowandb = \loworb=l$ then at most one
of the players can have an atom at $l$, so we are left
with two cases.

CASE I: If $\lowandb > \loworb$ or $\lowandb = \loworb$  but {\tt
AND} has no atom at $\lowandb$.  Consider the expected utility of
{\tt OR} when it bids $x\in[\loworb,\loworb+\epsilon]$ on item $2$
(with a low bid on the first item). {\tt OR} probability of winning
with such bids goes to zero as $\epsilon$ goes to zero. Therefore
{\tt OR}'s utility goes to zero, in contradiction to
Lemma~\ref{lem:or-positive}, which shows that {\tt OR} has a
strictly positive utility.

CASE II:  If $\lowandb < \loworb$ or $\lowandb = \loworb$  but {\tt
OR} has no atom at $\lowandb$.
Let us consider the set $B(\epsilon)$ of all the bids of {\tt AND}
which are coordinate-wise no larger than $(\lowanda +
\epsilon,\lowandb + \epsilon)$ as $\epsilon$ approaches zero.
First note that since $\lowanda>0$ and since by
Lemma~\ref{lem:or-low} {\tt OR} bids low on item $1$ with constant
probability, the expected payment of {\tt AND} on $B(\epsilon)$ is
some positive constant which is independent of $\epsilon$.
On the other hand we show that the probability that {\tt AND} wins
with a bid in $B(\epsilon)$  approaches $0$ as
  $\epsilon$ approaches
zero. This  implies that there exists and $\epsilon > 0$ such that
the {\tt AND} has negative utility for $B(\epsilon)$ and hence a
contradiction.

Consider the probability of {\tt AND} winning both items with  a bid
in $B(\epsilon)$. By Lemma \ref{lem:or-low} this is the sum of the
probability that {\tt AND} wins both items when {\tt OR} bids low on
item 1 and the probability that {\tt AND} wins both items when {\tt
OR} bids low on item 2. We
estimate these probabilities in the following cases.\\
 (IIa) If {\tt
OR} bids low on the first item, then the probability {\tt AND} wins
the second item  with a bid in $B(\epsilon)$  goes to zero with
$\epsilon$ by our assumption that $\lowandb < \loworb$ or $\lowandb
= \loworb$ but {\tt
OR} has no atom at $\lowandb$.\\
 (IIb)
If {\tt OR} bids low on the second item, then it cannot bid low on
the first item so it bids at least $\lowanda$ on item $1$.
% and may bid $\lowanda$ only if {\tt AND} has an atom there (since otherwise that bid looses with probability 1).
%
If {\tt AND} does not have an atom at $\lowanda$, then {\tt OR} bids
strictly above $\lowanda$ ($\lowanda$ is a low value), therefore,
{\tt AND}'s winning probability  with $B(\epsilon)$ goes to zero
with $\epsilon$.
If {\tt AND} has an atom at $\lowanda$, by
Lemma~\ref{lemma:atom_both}, {\tt OR} cannot have an atom at
$\lowanda$ so the probability that {\tt OR} bids less than $\lowanda
+ \epsilon$ goes to zero as $\epsilon$ approaches zero, and hence
the probability of {\tt AND} winning the first item goes too zero.
Hence, we have established that the probability that {\tt AND} wins
goes to zero as $\epsilon$ approaches zero.
 \qed\end{proof}

By Lemma \ref{lem:lowand}, any positive bid of {\tt OR} has a
positive probability of winning, hence the only low bids of {\tt OR}
is zero.

\begin{corollary} \label{axis}
{\tt OR} always bids either $(x,0)$ with $x>0$ or $(0,y)$ with
$y>0$.  Both of these events happen with positive probability.
\end{corollary}

\subsection{{\tt AND} Bids on the Diagonal}\label{and-diagonal}

In this section we are considering an equilibrium $(\bFand,\For)$ of
the {\tt AND}-{\tt OR} game, where $\bFand$ maximally correlated.

\begin{lemma} \label{diag}
$Pr_{(x,y)\sim \bFand} [x \ne y] = 0$.
\end{lemma}

\begin{proof}
By way of contradiction, without loss of generality assume
$Pr_{(x,y)\sim \Fand} [x < y] > 0$. So for some $0<b<c<H$ we will
have $Pr[{\tt AND}\:bids\: in\:[0,b)\times(c,H)] >0$ (since the
former event is the union over the countable choices of the latter
over all rationals $0<b<c<H$).  Now since $\bFand$ is maximally
correlated, either $\bFand(b,c)=\bFand(b,H)$ or
$\bFand(b,c)=\bFand(H,c)$.  However, the former possibility is in
contradiction to $Pr[{\tt AND}\:bids\: in\:[0,b)\times(c,H)] >0$ so
we must have $\bFand(b,c)=\bFand(H,c)$. This means that whenever
{\tt AND} bids at most $c$ on the second item, it also bids at most
$b$ on the first item.

In such a case, any bid $(0,y)$ for ${\tt OR}$, where $b<y<c$ is
strictly dominated the bid $(b,0)$ (since {\tt OR} wins at least
whenever $(0,y)$ wins, pays strictly less, and this happens with
positive probability since $\lowandb=0$). Therefore the {\tt OR}
never bids $(0,y)$ where $b<y<c$.

However {\tt AND} does bid at least $c$ on the second item with
positive probability. This contradicts Lemma~\ref{interval}.
\qed\end{proof}

\subsection{The exact forms}\label{exact}

 In this section we show that the previous lemma imply that an   equilibrium $(\bFand,\For)$
of the {\tt AND}-{\tt OR} game with $\bFand$ maximally correlated,
must have a specific form.

\begin{lemma} \label{or-form}
$\For(x,y) = \frac{1}{2} x/(1-x) + \frac{1}{2} y/(1-y)$ for $0\le
x,y\le 1/2$.
\end{lemma}

\begin{proof}
We know that {\tt OR} bids $0$ on exactly one of the items
(Corollary~\ref{axis}) so $\For(x,y)=\For(x,0)+\For(0,y)$. Let
$\alpha = \For(0,h)$  be the probability that {\tt OR} bids $0$ on
the first item, and let $1-\alpha = \For(h,0)$. Assume that {\tt
AND} bids $(x,y)$. With probability $\alpha$ {\tt OR} bids $0$ on
the first item and then 1) {\tt AND} wins the first item and pays
$x$ for it, and 2) {\tt AND} wins the second item with probability
$\For(0,y)$. Similarly, with probability $1-\alpha$ {\tt OR} bids
$0$ on the second item and then {\tt AND} 1) wins the second item
and pays $y$ for it, and 2) wins the first item with probability
$\For(x,0)$. So we conclude that $u_{{\tt AND}}(x,y)= \For(x,0)(1-x)
+ \For(0,y)(1-y) - \alpha x - (1-\alpha)y$.

 Now notice that $u_{{\tt AND}}(x,y)$ is the sum of a function
$g_1(x)=\For(x,0)(1-x)-\alpha x$, that depends only on $x$ and a
 function $g_2(y)=\For(0,y)(1-y)- (1-\alpha) y$ that depends only on $y$.
We claim that $g_1(x)$ must be a constant for all $x\in [0,h]$ and
$g_2(y)$  must be a constant for all $y\in [0,h]$. We prove this
claim for $g_1$, the proof for $g_2$ is the same. Assume for a
contradiction that $g_1(x_1)
> g_1(x_2)$ for some $x_1,x_2 \in [0,h]$. Then the bid $(x_1,x_2)$ of {\tt AND}
strictly dominates the bid $(x_2,x_2)$ in contradiction to the fact
that $(x_2,x_2)$  is a best response of {\tt AND} for $\For$ which
follows from Lemma~\ref{interval} (the support of {\tt AND} is an
interval), Lemma~\ref{lem:lowand} (the interval starts at $0$), and
Lemma~\ref{diag} ({\tt AND} bids on the diagonal).

 As {\tt OR}
does not have an atom at $(0,0)$, $\For(0,y)$ and $\For(x,0)$
approach 0 as $x$ and $y$ approach 0, respectively. Therefore
$u_{{\tt AND}}(x,y)$ approaches 0 as $x$ and $y$  approach 0.  It
follows that we must have   that $g_1(x)=\For(x,0)(1-x)-\alpha x =
0$ for all $0 \le x \le h$, and similarly  $g_2(y)=0$ for all $0 \le
y \le h$. I.e. $\For(x,0) = \alpha x/(1-x)$ and
$\For(0,y)=(1-\alpha)y/(1-y)$ for all $0 \le x,y \le h$. In
particular since $\For(h,h)=\For(0,h)+\For(h,0) = 1$, we have that
$\alpha h/(1-h)+(1-\alpha)h/(1-h)=h/(1-h)=1$ which implies that
$h=1/2$. But then substituting $h=1/2$ into the expression we get
$\For(0,h)=\For(0,1/2)=1-\alpha$. But by its definition
$\alpha=\For(0,h)$ so $\alpha=1/2$ and the lemma follows.
\qed\end{proof}

\begin{lemma} \label{and-form}
$\bFand(x,y)= \frac{v-1/2}{v-min(x,y)}$.
\end{lemma}

\begin{proof}
Since {\tt AND} bids on the diagonal, clearly $\bFand(x,y) =
\bFand(min(x,y),min(x,y))$, so it suffices to characterize
$\bFand(x,x)$ for all $x$. The utility for {\tt OR} for bidding
$(x,0)$ is $(v - x)  \bFand(x,x)$.  By Corollary \ref{axis} (lowest
bid of {\tt OR} is $0$), Corollary \ref{cor:highestbid} ({\tt OR}
and {\tt AND} have the same highest bid),  Lemma \ref{or-form} (this
highest bid is $1/2$), and Lemma \ref{interval} ({\tt OR} bids in an
interval) we know that the support of {\tt OR} is the interval
$(0,1/2)$. So for every $0 < x < 1/2$, $(x,0)$ is a best response to
{\tt AND} and thus $(v - x) \bFand(x,x)$ is a constant independent
of $x$. For $x=1/2$ we know that $\bFand(1/2,1/2)=1$ and thus this
constant is $v-1/2$. It follows that $(v - x)  \bFand(x,x) = v-1/2$
for all $x$, i.e., $\bFand(x,x) = (v-1/2)/(v-x)$, and the lemma
follows. \qed\end{proof}

\subsection{Completing the Proof} \label{thm-pf}

Now let us put everything together to prove Theorem~\ref{thm:main}.

\begin{proofof}{Theorem~\ref{thm:main}}
We start by showing the necessary conditions for an equilibrium.
Take an equilibrium $(\Fand,\For)$.  As Lemma~\ref{eqtoeq} shows
$(\bFand,\For)$ is also an equilibrium, with the same allocations,
payments, and  utilities. From Lemmas \ref{or-form} and
\ref{and-form} we have that $\For(x,y) = (x/(1-x) + y/(1-y))/2$ and
$\bFand(x,y)= \min(\Fand(x,H),\Fand(H,y)) = (v-1/2)/(v-\min(x,y))$.
By Lemma~\ref{weak-dom}, $\Fand$ is dominated by $\bFand$.

To show that $\Fand(0,0)=\bFand(0,0)$, assume for a contradiction
that $\Fand(0,0)\neq \bFand(0,0)$. Since by definition $\bFand(x,y)
\ge \Fand(x,y)$, we have that $\Fand(0,0)<\bFand(0,0)=(v-1/2)/v$.
For   $x,y >0$ the utility of {\tt OR} from the bid $(x,y)$ is
$u_{or}(x,y) = \Fand(x,H)(v-x) + \Fand(H,y)(v-y) - \Fand(x,y)\cdot
v$ (for $x=0$ or $y=0$ the atom of $\Fand$ at $x=0$ and $y=0$  may
cause the utility to be lower depending on the tie breaking rule).
Now let $x$ and $y$ approach $0$ and we get at the limit a utility
of $v \cdot (\Fand(0,H) + \Fand(H,0) - \Fand(0,0)) > v \cdot
(\bFand(0,H) + \bFand(H,0) - \bFand(0,0)) = v \cdot \bFand(0,0) =
v-1/2$ (where the inequality follows since the marginal
distributions of $\Fand$ and $\bFand$ are the same). Thus for small
enough $x$ and $y$ the utility of {\tt OR} is strictly greater than
$v-1/2$. This however contradicts our derivation of $\For$ whose
support includes the bid $(0,1/2)$ for which $u_{or}(0,1/2) =
v-1/2$, a contradiction to it being a best response to $\Fand$.
Therefore, $\Fand(0,0)=\bFand(0,0)$.

\medskip

\ignore{ It remains to show that $\Fand(x,0)=\Fand(0,0)=\Fand(0,y)$,
which follows from the fact that it is an equilibrium. For
contradiction assume that $\Fand(x,0)>\Fand(0,0)$ (the other case is
similar). Let $q_x$ be the probability that {\tt AND} wins both
items with bid $(x,0)$. If $q_x=q_0$ then
%the {\tt AND} has a negative utility for
any bid $(z,0)$, where $z\in(0,x]$, is strictly dominated by
$(0,0)$, since it has a strictly higher payment since the low value
of {\tt OR} is $0$. If $q_x > q_0$ then the {\tt OR} can deviate and
always bid at least $\epsilon$ on item $2$, which will increase its
payment by at most $\epsilon$ and increase its probability of
winning by at least $q_x-q_0>0$. Therefore we reached a
contradiction.}

We now show the sufficient conditions for an equilibrium. For the
fixed distributions $\bFand(x,y)= (v-1/2)/(v-\min(x,y))$ and
$\For(x,y) = (x/(1-x) + y/(1-y))/2$ one may directly verify that
$(\bFand,\For)$ is an equilibrium (as was shown in
\cite{HassidimKMN11}). Now take $\Fand$ such that $\bFand(x,y)=
(v-1/2)/(v-\min(x,y))$, $\Fand(0,0)=(v-1/2)/v$, and
$\Fand(x,0)=\Fand(0,0)=\Fand(0,y)$.\footnote{This follows since
$\Fand^*(0,0) = \Fand(0,0) \leq \Fand(0,y) \leq \Fand(0,H) =
\Fand^*(0,H) = \Fand^*(0,0)$.} We first need to show that $\Fand$ is
also a best response to $\For$ (and not just $\bFand$), that is we
need to show that $u_{and}(\Fand,\For)=u_{and}(\bFand,\For)$.  By
Proposition~\ref{better} this will happen whenever the probability
that {\tt AND} wins no item is the same in both cases. Since
$\Fand(x,0)=\Fand(0,0)=\Fand(0,y)$, the probability for {\tt AND}
winning no item is exactly the probability that it bids $(0,0)$
times the probability that {\tt OR} wins the tie at its $0$ bid,
which is the same in $\Fand$ and $\bFand$ since
$\Fand(0,0)=\bFand(0,0)$.

Finally we need to show that $\For$ is also a best response to
$\Fand$. As before, for  $0 < x,y$, the utility of {\tt OR} from a
bid $(x,y)$ is $u_{or}(x,y)=\Fand(x,H)(v-x) + \Fand(H,y)(v-y)
-\Fand(x,y)v$. Notice that $\Fand(x,H)=\bFand(x,H) = (v-1/2)/(v-x)$
so the first and the second terms equal the constant $v-1/2$. It
follows that the maximum utility is obtained as $x$ and $y$ approach
$0$ (since this minimizes the last term $F(x,y)$). The utility of
{\tt OR} when $x$ and $y$ approach $0$, approaches (from below)
$2v-1-\Fand(0,0)v = 2v-1-\bFand(0,0)v = (2v-1)-(v-1/2)=v-1/2$.
Since $\Fand(x,0)=\Fand(0,0)=\Fand(0,y)$, we have that the utility
of {\tt OR}  is $v-1/2$ at any point in the support of $\For^*$, and
hence it is a best response.
%
%But $u_{or}(\Fand,\For) = u_{or}(\bFand,\For)$ (Lemma~\ref{eqtoeq}) which we evaluate as
%$u_{or}(\bFand,\For) = v-1/2$, e.g. by looking at the bid $(0,1/2)$
%which in the support of {\tt OR} in the equilibrium $(\bFand,\For)$.
%\qed
\end{proofof}

\section{Properties of the  equilibrium}

We present few properties of the Nash equilibrium  in Theorem
\ref{thm:main}. Our analysis is a function of the value $v$ (of the
{\tt OR} player). We analyze the probability that each player wins,
the expected revenue and the expected social welfare. By Theorem
\ref{thm:main} all these quantities are identical in every Nash
equilibrium. (The proofs and the figures are in the Appendix.)

First, we derive the probability that the {\tt AND} player wins
(clearly the probability that the {\tt OR} player wins is the
complement). This probability is depicted in Figure
\ref{fig-and-wins}.

\begin{lemma}
\label{lemma:prob-and}
 The probability that the {\tt AND} player
wins is $\frac{\ln(2) - \frac{1}{2}}{v} +
O\left(\frac{1}{v^2}\right) $ and for $v=1$ this probability is
$1/4$.
\end{lemma}

\ignore{
\begin{proof}
If $v\not = 1$ we get that
\begin{eqnarray*}
\Pr[{\tt AND}\; wins] & = & \int_{0}^{1/2} \Fand'(x) \For(x) dx \\
& = & \int_{0}^{1/2} \frac{v-\frac{1}{2}}{(v-x)^2} \frac{x}{1-x} dx \\
& = & \left(v-\frac{1}{2} \right)\left[\frac{\ln \frac{v-x}{1-x}}{(v-1)^2} - \frac{v}{(v-1)(v-x)} \right]_0^{1/2} \\
& = & \left(v-\frac{1}{2} \right)\left[ \frac{\ln(2v-1)}{(v-1)^2} - \frac{v}{(v-1)(v-1/2)} - \frac{\ln v}{(v-1)^2} + \frac{1}{v-1} \right] \\
& = & \left(v-\frac{1}{2} \right) \left[ \frac{\ln(2v-1) - \ln v}{(v-1)^2} - \frac{1/2}{(v-1)(v-1/2)} \right] \\
& = & \frac{(v-{1/2})\ln(2-\frac{1}{v}) - \frac{1}{2}(v-1)}{(v-1)^2} \\
& = & \frac{\ln(2) - \frac{1}{2}}{v} + O\left(\frac{1}{v^2}\right)
\end{eqnarray*}

For $v=1$ a similar calculation shows that
$$
\Pr[{\tt AND}\; wins \mid v=1] = \frac{1}{2} \int_0^{1/2}
\frac{x}{(1-x)^3} dx = \frac{1}{2}\left[ \frac{2x-1}{2(x-1)^2}
\right]_0^{1/2} = \frac{1}{4} \ .
$$
\qed\end{proof} }

 Next we compute the expected revenue. The expected utility of the {\tt AND}
player is $0$ and therefore the revenue from the {\tt AND} player
equals to the probability that it wins. It remains to compute the
revenue from the {\tt OR} player.

\begin{theorem}
\label{thm:rev-OR}
 The expected revenue from the {\tt OR} player is
$1-\ln 2 -O(\frac{1}{v})$. For $v=1$ the expected revenue from the
{\tt OR} player is $1/4$.
\end{theorem}

\ignore{
\begin{proof}
\begin{eqnarray*}
Revenue({\tt OR}) & = & \int_{0}^{1/2}x \For'(x)\Fand(x)dx \\
& = & \int_0^{1/2} x \frac{1}{(1-x)^2}\frac{(v-\frac{1}{2})}{(v-x)} dx \\
& = & \frac{(v-\frac{1}{2})}{(v-1)^2} \left[ v \ln \left(\frac{1-x}{v-x} \right) + \frac{(v-1)}{(1-x)} \right]_0^{1/2}  \\
& = & \frac{(v-\frac{1}{2})}{(v-1)^2} \left[v \ln \left(
\frac{\frac{1}{2}}{v-\frac{1}{2}} \right) + \frac{(v-1)}{1/2}
-\left(v\ln\frac{1}{v} + (v-1) \right) \right] \\
& = & \frac{(v-\frac{1}{2})}{(v-1)^2} \left[v-1 - v\ln 2 + v\ln \frac{v}{v-\frac{1}{2}}  \right] \\
& = & \frac{(v-\frac{1}{2})}{(v-1)^2} \left[v-1 - v\ln 2  + v\ln(1 + \frac{\frac{1}{2}}{v-\frac{1}{2}}) \right] \\
& = & \frac{(v-\frac{1}{2})}{(v-1)^2} \left[v-1 - (v-1)\ln 2 - \ln 2 + v \ln\left(1+\frac{1}{2v-1}\right) \right] \\
& = & (1-\ln 2) \frac{v-\frac{1}{2}}{v-1} - O\left(\frac{1}{v}
\right) \\
&=& 1-\ln 2 - O\left(\frac{1}{v} \right) =0.3068-
 O\left(\frac{1}{v} \right)
\end{eqnarray*}

For $v=1$ a similar calculation shows that the revenue of the {\tt
OR} player is $0.25$.
\qed
\end{proof}
}

 The revenue from the {\tt OR} player is plotted in Figure
\ref{revenue-OR} as a function of the value $v$. The revenue from
the auction (i.e., sum of both players) is shown in Figure
\ref{revenue-OR-AND}.
%\end{proof}

Using the probability that each  player  wins, we can compute the
expected social welfare, which is $(\Pr[{\tt And}\; wins] + v\cdot
\Pr[{\tt OR}\; wins])$.

\begin{theorem}
The expected social welfare is $v-\ln(2)+1/2+\frac{\ln(2)-1/2}{v}
+O(1/v^2)$.
%For $v >> 1$ the social welfare is about $v-\ln(2)+1/2$.
\end{theorem}

Figure \ref{fig:poa} shows the Price of Anarchy of the equilibrium.
That is we divide the expected social welfare in  equilibrium
$(\Pr[{\tt And}\; wins] + v\cdot \Pr[{\tt OR}\; wins])$ by the
maximum  social welfare, that is $\max\{v,1 \}$. The difference
$\max\{v,1\} - (\Pr[{\tt And}\; wins] + v\cdot \Pr[{\tt OR}\;
wins])$ is shown in Figure \ref{fig:loss}. The expected loss
converges to $\ln(2) - 0.5 \approx 0.19$ as the value $v$ of the
{\tt OR} player goes to infinity.

\bibliographystyle{plain}
\bibliography{firstprice}

%\newpage
\appendix

\section{Acknowledgements}
Research by Avinatan Hassidim was supported in part
 by a grant from the Israel
Science Foundation (ISF), and by a grant from the German Israel Foundation.

Research by Haim Kaplan was supported in part
 by a grant from the Israel
Science Foundation (ISF), by a grant from United States-Israel
Binational Science Foundation (BSF), by The Israeli Centers of
Research Excellence (I-CORE) program, (Center  No. 4/11), and by the
Google Inter-university center for Electronic Markets and Auctions.

Research by Yishay Mansour was supported in part by a grant from the
the Science Foundation (ISF), by a grant from United States-Israel
Binational Science Foundation (BSF), by a grant from the Israeli
Ministry of Science (MoS),  by The Israeli Centers of Research
Excellence (I-CORE) program, (Center No. 4/11) and by the Google
Inter-university center for Electronic Markets and Auctions.

Research by Noam Nisan was supported  by a grant from the Israeli
Science Foundation (ISF), by The Israeli Centers of Research
Excellence (I-CORE) program, (Center No. 4/11) and by the Google
Inter-university center for Electronic Markets and Auctions.

\section{Missing Proofs}

\begin{proofof}{Lemma~\ref{lemma:prob-and}}
If $v\not = 1$ we get that
\begin{eqnarray*}
\Pr[{\tt AND}\; wins] & = & \int_{0}^{1/2} \Fand'(x) \For(x) dx \\
& = & \int_{0}^{1/2} \frac{v-\frac{1}{2}}{(v-x)^2} \frac{x}{1-x} dx \\
& = & \left(v-\frac{1}{2} \right)\left[\frac{\ln \frac{v-x}{1-x}}{(v-1)^2} - \frac{v}{(v-1)(v-x)} \right]_0^{1/2} \\
& = & \left(v-\frac{1}{2} \right)\left[ \frac{\ln(2v-1)}{(v-1)^2} - \frac{v}{(v-1)(v-1/2)} - \frac{\ln v}{(v-1)^2} + \frac{1}{v-1} \right] \\
& = & \left(v-\frac{1}{2} \right) \left[ \frac{\ln(2v-1) - \ln v}{(v-1)^2} - \frac{1/2}{(v-1)(v-1/2)} \right] \\
& = & \frac{(v-{1/2})\ln(2-\frac{1}{v}) - \frac{1}{2}(v-1)}{(v-1)^2} \\
& = & \frac{\ln(2) - \frac{1}{2}}{v} + O\left(\frac{1}{v^2}\right)
\end{eqnarray*}

For $v=1$ a similar calculation shows that
$$
\Pr[{\tt AND}\; wins \mid v=1] = \frac{1}{2} \int_0^{1/2}
\frac{x}{(1-x)^3} dx = \frac{1}{2}\left[ \frac{2x-1}{2(x-1)^2}
\right]_0^{1/2} = \frac{1}{4} \ .
$$
\end{proofof}

\begin{proofof}{Theorem~\ref{thm:rev-OR}}
\begin{eqnarray*}
Revenue({\tt OR}) & = & \int_{0}^{1/2}x \For'(x)\Fand(x)dx \\
& = & \int_0^{1/2} x \frac{1}{(1-x)^2}\frac{(v-\frac{1}{2})}{(v-x)} dx \\
& = & \frac{(v-\frac{1}{2})}{(v-1)^2} \left[ v \ln \left(\frac{1-x}{v-x} \right) + \frac{(v-1)}{(1-x)} \right]_0^{1/2}  \\
& = & \frac{(v-\frac{1}{2})}{(v-1)^2} \left[v \ln \left(
\frac{\frac{1}{2}}{v-\frac{1}{2}} \right) + \frac{(v-1)}{1/2}
-\left(v\ln\frac{1}{v} + (v-1) \right) \right] \\
& = & \frac{(v-\frac{1}{2})}{(v-1)^2} \left[v-1 - v\ln 2 + v\ln \frac{v}{v-\frac{1}{2}}  \right] \\
& = & \frac{(v-\frac{1}{2})}{(v-1)^2} \left[v-1 - v\ln 2  + v\ln(1 + \frac{\frac{1}{2}}{v-\frac{1}{2}}) \right] \\
& = & \frac{(v-\frac{1}{2})}{(v-1)^2} \left[v-1 - (v-1)\ln 2 - \ln 2 + v \ln\left(1+\frac{1}{2v-1}\right) \right] \\
& = & (1-\ln 2) \frac{v-\frac{1}{2}}{v-1} - O\left(\frac{1}{v}
\right) \\
&=& 1-\ln 2 - O\left(\frac{1}{v} \right) =0.3068-
 O\left(\frac{1}{v} \right)
\end{eqnarray*}

For $v=1$ a similar calculation shows that the revenue of the {\tt
OR} player is $0.25$.
\end{proofof}

\section{Missing Figures}

In this section we present the missing figures. All the figures depict properties of the equilibrium, as a function
of the value of the {\tt OR} player; the value of the {\tt AND} player is taken to be $1$.
The graphs are only presented for $v > 0.5$, as otherwise there is a Walrasian equilibrium.
In all the figures the asterisk depicts the crossover point, in which the value of the {\tt OR}
player is 1.

\begin{figure}[htbp]
\centerline{ \epsfig{figure=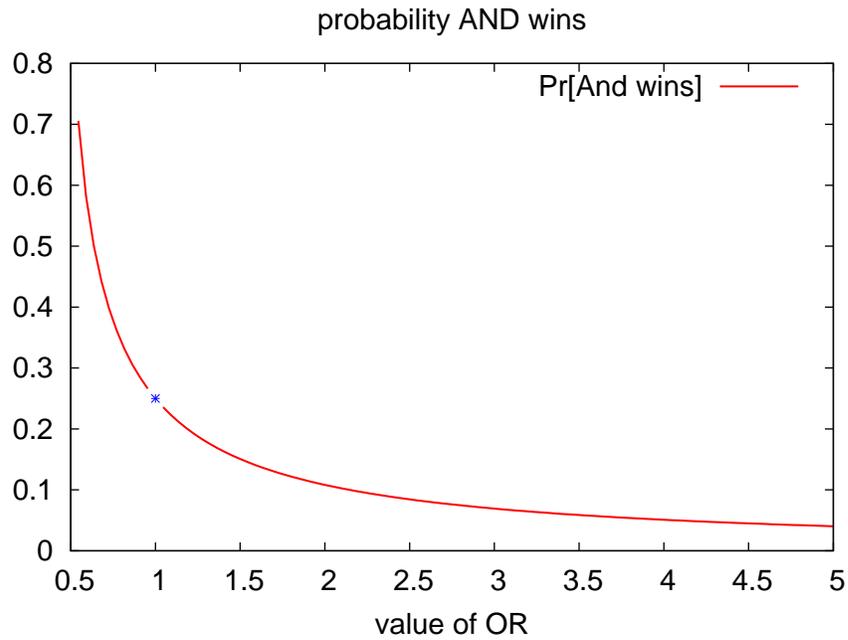} } \caption{The probability
that the {\tt AND} player wins which is the same as the revenue
generated from the {\tt AND} player. \label{fig-and-wins}}
\end{figure}

\begin{figure}[htbp]
\centerline{ \epsfig{figure=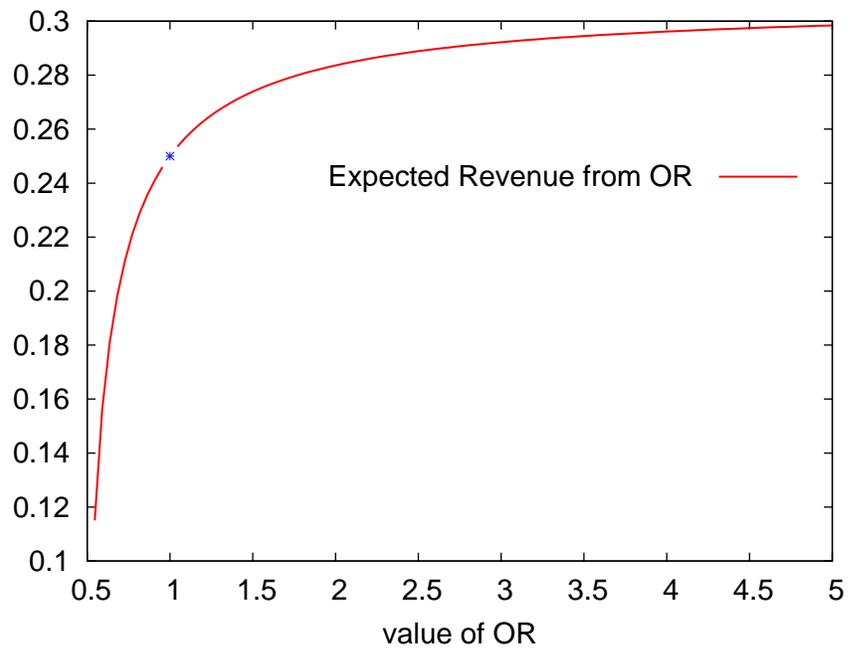} } \caption{The revenue
from the {\tt OR} player. \label{revenue-OR}}
\end{figure}
\begin{figure}[htbp]
\centerline{ \epsfig{figure=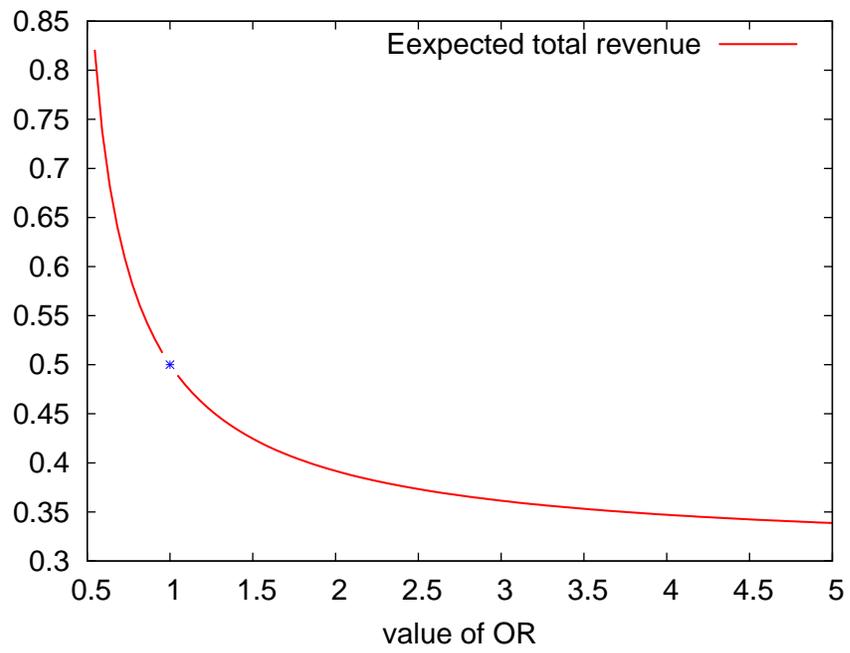} } \caption{The revenue
from the {\tt OR} player and the {\tt AND} player.
\label{revenue-OR-AND}}
\end{figure}

\begin{figure}[htbp]
\centerline{ \epsfig{figure=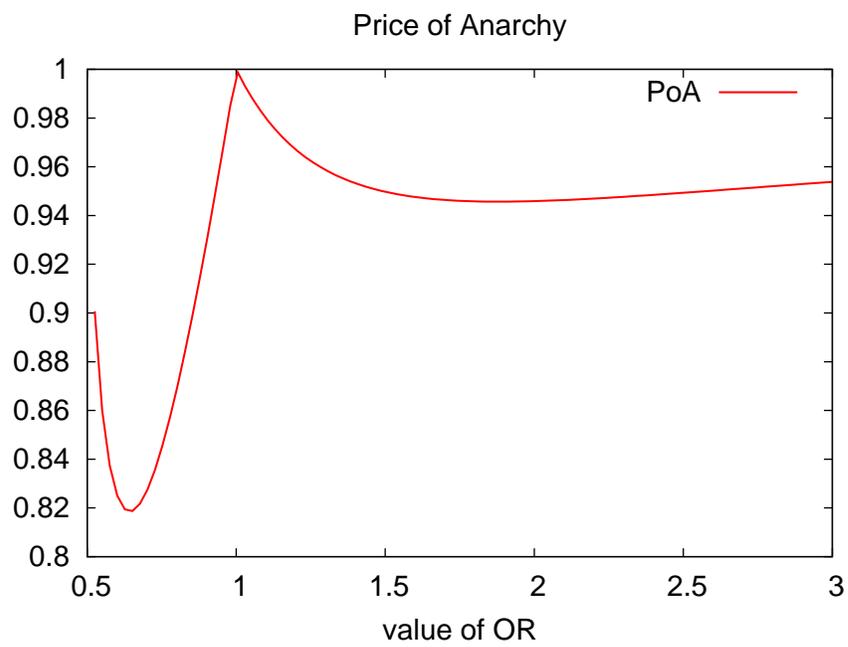} } \caption{The Price of Anarchy
of the {\tt AND}/{\tt OR} game. This is the fraction of the optimal social welfare obtained by the auction. For $v < 1$ it achieves a minimum of
$\approx 0.818485$ at $v \approx 0.643028$. For $v > 1$ it achieves
a minimum of $\approx 0.945682$ at $v \approx 1.87999$.
 \label{fig:poa}}
\end{figure}

\begin{figure}[htbp]
\centerline{ \epsfig{figure=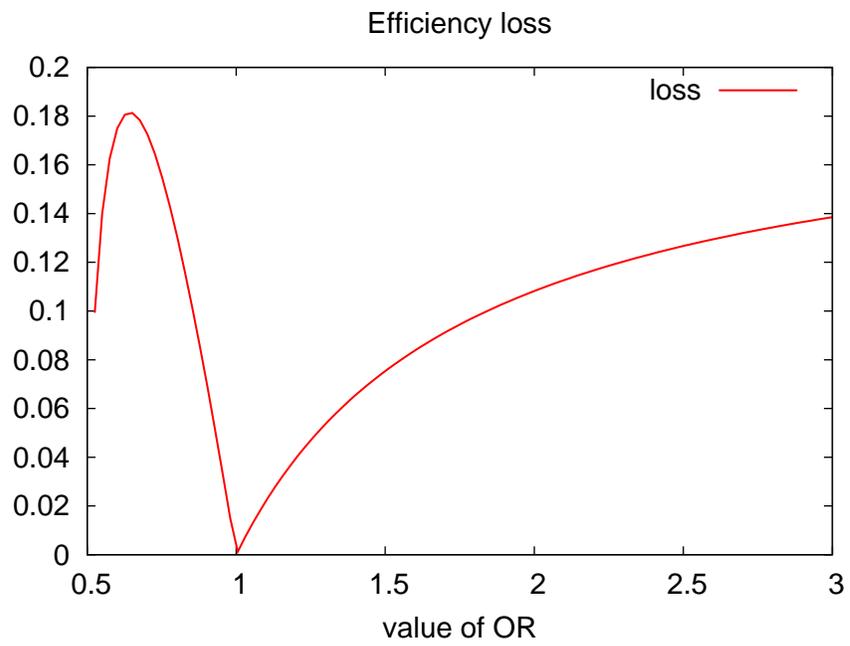} } \caption{The
additive loss in social welfare of the Nash equilibrium of the {\tt
AND}/{\tt OR} game. That is, the optimal social welfare minus the social welfare of the Nash equilibrium.
 \label{fig:loss}}
\end{figure}
\end{document}